\documentclass[preprint]{elsarticle}
\journal{Information and computation}
\makeatletter
\def\ps@pprintTitle{%
 \let\@oddhead\@empty
 \let\@evenhead\@empty
 \def\@oddfoot{\centerline{\thepage}}%
 \let\@evenfoot\@oddfoot}
\makeatother
\usepackage{amsmath,amssymb,amsthm,nicefrac,relsize,graphicx,paralist}
\usepackage[left=3.2cm,right=3.2cm,top=2.8cm,bottom=2.8cm]{geometry}

\newtheorem{theorem}{Theorem}[section]
\newtheorem{proposition}[theorem]{Proposition}
\newtheorem{corollary}[theorem]{Corollary}
\newtheorem{lemma}[theorem]{Lemma}

\theoremstyle{definition}
\newtheorem{definition}[theorem]{Definition}
\newtheorem{example}[theorem]{Example}

\theoremstyle{remark}
\newtheorem{remark}[theorem]{Remark}

\newcommand{\displayproblem}[3]
{%
  \begin{trivlist}%
  \item
  \begin{list}{}{\setlength{\leftmargin}{5.5em}\setlength{\rightmargin}{1em}\setlength{\labelwidth}{\leftmargin}}%
  \item[\bf Problem:] {\sc #1}
  \item[\bf Instance:] #2
  \item[\bf Question:] #3
  \end{list}%
  \end{trivlist}%
}

\newcommand{\sml}[1]{\mathsmaller{ #1 }}
\newcommand{\norm}[1]{{\lVert #1 \rVert}_{1}}
\newcommand{\pnorm}[1]{{\lVert #1 \rVert}_{p}}
\newcommand{\mysum}{\sum}
\newcommand{\supp}{\operatorname{supp}}
\newcommand{\arginf}{\operatorname{arg\,inf}}
\newcommand{\argsup}{\operatorname{arg\,sup}}

\renewcommand{\qed}{\hfill $ \blacksquare $}

\begin{document}

\begin{frontmatter}

\title{On the Entropy of Couplings\tnoteref{t1}}

\tnotetext[t1]{Part of this work was presented at the 2012 IEEE Information Theory Workshop (ITW) \cite{kovacevic}.
               The work was supported by the Ministry of Education, Science and 
               Technological Development of the Republic of Serbia 
               (grants TR32040 and III44003).}

\author[ftn]{Mladen~Kova\v{c}evi\'{c} \corref{cor}}
\cortext[cor]{Corresponding author (kmladen@uns.ac.rs).}

\author[ftn]{Ivan~Stanojevi\'{c}}

\author[ftn]{Vojin~\v{S}enk}

\address[ftn]{Department of Electrical Engineering, 
              University of Novi Sad,\\
              Trg Dositeja Obradovi\'{c}a 6, 21000 Novi Sad, Serbia}



\begin{abstract}
In this paper, some general properties of Shannon information
measures are investigated over sets of probability distributions
with restricted marginals.
Certain optimization problems associated with these functionals are
shown to be NP-hard, and their special cases are found to be essentially
information-theoretic restatements of well-known computational problems,
such as the \textsc{Subset sum} and the \textsc{3-Partition}.
The notion of minimum entropy coupling is introduced and its relevance
is demonstrated in information-theoretic, computational, and
statistical contexts.
Finally, a family of pseudometrics (on the space of discrete
probability distributions) defined by these couplings is studied,
in particular their relation to the total variation distance, and
a new characterization of the conditional entropy is given.
\end{abstract}


\end{frontmatter}

\section{Introduction}

Distributions with fixed marginals have been studied extensively in the
probability literature (see for example \cite{fixed} and the references
therein).
They are closely related to (and sometimes identified with, as will be
the case in this paper) the concept of coupling, which has proven to be
a very useful proof technique in probability theory \cite{coupling}, and
in particular in the theory of Markov chains \cite{peres}.
In statistics, a related notion of contingency tables is of considerable
importance \cite{csiszar2}.
There is also rich literature on the geometrical and combinatorial
properties of sets of distributions with given marginals, which are known
as transportation polytopes in this context (see, e.g., \cite{brualdi}).
We investigate here these objects from a certain information-theoretic
perspective.
Our results and the general outline of the paper are briefly described
below.

In Section \ref{preliminaries} we recall the definitions and elementary
properties of the quantities under study, namely, information-theoretic
functionals and couplings.
Notational conventions are also introduced here.

In Section \ref{optimization} we discuss properties of Shannon
information measures under constraints on the marginal distributions.
In particular, certain optimization problems associated with these
functionals are studied.
Most of them are, in a sense, the reverses of the well-known optimization
problems, such as the maximum entropy principle, channel capacity, and
information projections.
The general problems of entropy minimization, maximization of mutual
information, and maximization of information divergence are all shown
to be intractable.
Since mutual information is a good measure of dependence of two random
variables, this will also lead to a similar result for all measures of
dependence satisfying R\'{e}nyi's axioms, and to a statistical scenario
where this result might be of interest.
Furthermore, these problems are found to be basically information-theoretic
restatements of some well-known problems in complexity theory.
The infinite-alphabet case is also discussed in this section, in
particular the questions of continuity and existence of extrema.

In Section \ref{distance} we define a family of (pseudo)metrics on the
space of probability distributions, that is based on the so-called
minimum entropy coupling in the same way as the total variation distance
is based on the maximal coupling.
The relation between these distances is derived from Fano's inequality.
Other properties of the new metrics are also discussed, in particular an
interesting characterization of the conditional entropy that they yield.

\section{Preliminaries}
\label{preliminaries}

This section summarizes the definitions and elementary properties of the
notions used in the sequel.
Some conventions that we adopt are as follows.
All random variables are assumed to be discrete, with alphabet $ \mathbb{N} $
-- the set of positive integers, or a subset of $ \mathbb{N} $ of the form
$ \{ 1, \ldots, n \} $.
The base of the logarithm $ \log $ is assumed to be $ 2 $, though this will
be relevant only in the statement of the Pinsker-Csisz\'ar-Kemperman inequality
\eqref{pinsker}.
For a probability distribution $ P = (p_i) $, we denote its support by
$ \supp(P) = \{ i : p_i > 0 \} $.
The size of the support is denoted by either $ |\supp(P)| $ or simply $ |P| $.
We will sometimes write $ P(i) $ for the masses of $ P $.

\subsection{Shannon information measures}
\label{shannon}

Shannon entropy of a random variable $ X $ with probability distribution
$ P = (p_i) $ is defined as
\begin{equation}
  H(X) \equiv H(P) = -\mysum_{i} p_i \log p_i
\end{equation}
with the usual convention $ 0 \log 0 = 0 $ being understood.
$ H $ is a strictly concave functional in $ P $ \cite{cover}.
Further, for a pair of random variables $(X,Y)$ with joint distribution
$ S = (s_{i,j}) $ and respective marginal distributions $ P = (p_i) $ and
$ Q = (q_j) $, the following defines their joint entropy
\begin{equation}
  H(X,Y) \equiv H_{X,Y}(S) = -\mysum_{i,j} s_{i,j} \log s_{i,j} ,
\end{equation}
conditional entropy
\begin{equation}
  H(X|Y) \equiv H_{X|Y}(S) = -\mysum_{i,j} s_{i,j} \log \frac{ s_{i,j} }{ q_j } ,
\end{equation}
and mutual information
\begin{equation}
  I(X;Y) \equiv I_{X;Y}(S) = \mysum_{i,j} s_{i,j} \log \frac{ s_{i,j} }{ p_{i} q_{j} } ,
\end{equation}
again with appropriate conventions.
The above quantities, usually referred to as the Shannon information measures
\cite{shannon1}, are all related by simple identities
\begin{equation}
\label{identity}
 \begin{aligned}
   H(X,Y) &= H(X) + H(Y) - I(X;Y)  \\ 
          &= H(X) + H(Y|X) 
 \end{aligned} 
\end{equation}
and obey the following inequalities
\begin{align}
\label{ineqH}
  \max \big\{ H(X), H(Y) \big\}  \leq H(X&,Y) \leq H(X) + H(Y) , \\
\label{ineqI}
  \min \big\{ H(X), H(Y) \big\} &\geq I(X;Y) \geq 0 , \\
\label{ineqHx}
  0 \leq H(X|Y) &\leq H(X) .
\end{align}
The equalities on the right-hand sides of \eqref{ineqH}--\eqref{ineqHx} are
attained if and only if $ X $ and $ Y $ are independent.
The equalities on the left-hand sides of \eqref{ineqH} and \eqref{ineqI} are
attained if and only if $ X $ deterministically depends on $ Y $ (i.e., iff
$ X $ is a function of $ Y $), or vice versa.
The equality on the left-hand side of \eqref{ineqHx} holds if and only if $ X $
deterministically depends on $ Y $.
We will use some of these properties in our proofs; for their demonstration
we point the reader to the standard reference \cite{cover}. 

From identities \eqref{identity} one immediately observes the following: 
Over a set of bivariate probability distributions with fixed marginals
(and hence fixed marginal entropies $ H(X) $ and $ H(Y) $), all the above
functionals differ up to an additive constant (and a minus sign in the case
of mutual information), and hence one can focus on studying only one of them
and easily translate the results for the others.
This fact will also be exploited later.

Relative entropy (information divergence, Kullback-Leibler divergence) of
the distribution $ P $ with respect to the distribution $ Q $ is the following
functional
\begin{equation}
  D(P||Q) = \mysum_{i} p_i \log \frac{ p_i }{ q_i } ,
\end{equation}
where $ 0 \log \frac{0}{q} = 0 $ and $ p \log \frac{p}{0} = \infty $ for every
$ q \geq 0 $, $ p > 0 $.

\subsection{Couplings of probability distributions}
\label{transport}

A coupling of two probability distributions $ P $ and $ Q $ is a bivariate
distribution $ S $ (on the product space, in our case $ \mathbb{N}^2 $) with
marginals $ P $ and $ Q $.
This concept can also be defined for random variables in a similar manner \cite{coupling}.

Let $ \Gamma_{n}^\sml{(1)} $ and $ \Gamma_{n\times m}^\sml{(2)} $ denote the
sets of one- and two-dimensional probability distributions with alphabets
of size $ n $ and $ n \times m $, respectively
\begin{align}
  \Gamma_{n}^\sml{(1)} 
     &= \Bigg\{ (p_i) \in \mathbb{R}^{n} \,:\,
                 p_i \geq 0 \,,\, \mysum_{i} p_i = 1 \Bigg\}  \\
  \Gamma_{n\times m}^\sml{(2)} 
     &= \Bigg\{ (p_{i,j}) \in \mathbb{R}^{n\times m} \,:\,
                 p_{i,j} \geq 0 \,,\, \mysum_{i,j} p_{i,j} = 1 \Bigg\}
\end{align}
and let $ \mathcal{C}(P,Q) $ denote the set of all couplings of
$ P \in \Gamma_{n}^\sml{(1)} $ and $ Q \in \Gamma_{m}^\sml{(1)} $
\begin{equation}
  \mathcal{C}(P,Q)
     = \Bigg\{ S \in \Gamma_{n\times m}^\sml{(2)} \,:\, 
              \mysum_{j} s_{i,j} = p_i \,,\, \mysum_{i} s_{i,j} = q_j \Bigg\} .
\end{equation}
The sets $ {\mathcal C}(P,Q) $ are restrictions to $ \mathbb{R}_{+}^{n\times m} $
of parallel affine $ (|P| - 1)(|Q| - 1) $-dimensional subspaces of $ \mathbb{R}^{n \times m} $.
They are convex, compact, and form a partition of $ \Gamma_{n\times m}^\sml{(2)} $.

The set of distributions with fixed marginals is basically the set of
matrices with nonnegative entries and prescribed row and column sums.
Such sets are special cases of the so-called transportation polytopes
\cite{brualdi}.

We will also find it interesting to study information measures over the
sets of distributions whose one marginal and the support of the other are
fixed
\begin{equation}
  {\mathcal C}(P,m) = \bigcup_{ Q \in \Gamma_{m}^\sml{(1)} } {\mathcal C}(P,Q) .
\end{equation}
These sets are also convex polytopes and form a partition of
$ \Gamma_{n\times m}^\sml{(2)} $ when $ P $ varies through $ \Gamma_{n}^\sml{(1)} $.

\section{Information measures and couplings}
\label{optimization}

In the following we analyze some general properties of Shannon information
measures, as well as natural optimization problems associated with these
functionals, over domains of the form $ {\mathcal C}(P,Q) $ and $ {\mathcal C}(P,m) $.
The proofs presented are not difficult, but they have a number of important
consequences, as discussed in Section \ref{generalizations}.
Some closely related problems over $ {\mathcal C}(P,Q) $, in the context of
computing the metric $ \underline{\Delta}_1(P,Q) $ (defined in Section
\ref{distance}), are also studied in \cite{vidyasagar}.

\subsection{Optimization over $ {\mathcal C}(P,Q) $}
\label{bothmarginals}

Due to \eqref{identity}, we can focus on the optimization of $ H_{X,Y} $ only.
In regard to this, we introduce the following definition, whose relevance will
be demonstrated throughout this and the following section.

\begin{definition}
 \emph{Minimum entropy coupling} of probability distributions $ P $ and $ Q $
is a bivariate distribution $ S^* \in {\mathcal C}(P,Q) $ that minimizes
the entropy functional $ H \equiv H_{X,Y} $, i.e., 
\begin{equation}
  H(S^*) = \inf_{ S \in {\mathcal C}(P,Q) } H(S) .
\end{equation}
\end{definition}

Note that the maximization of entropy over $ {\mathcal C}(P,Q) $ is trivial --
the maximizer is always $ P \times Q = (p_i q_j) $.

Minimum entropy couplings exist for any $ P \in \Gamma_{n}^\sml{(1)} $ and
$ Q \in \Gamma_{m}^\sml{(1)} $ because sets $ {\mathcal C}(P,Q) $ are compact
(closed and bounded) and entropy is continuous over $ \Gamma_{n\times m}^\sml{(2)} $
and hence attains its extrema.
Note, however, that they need not be unique.
From the strict concavity of entropy one concludes that the minimum entropy
couplings must be vertices of the polytope $ {\mathcal C}(P,Q) $ (i.e., they
cannot be expressed as $ \lambda S + (1 - \lambda)T $, with
$ S, T \in {\mathcal C}(P,Q) $, $ \lambda \in (0,1) $).
Finally, from identities \eqref{identity} it follows that the minimizers of
$ H_{X,Y} $ over $ {\mathcal C}(P,Q) $ are simultaneously the minimizers of
$ H_{X|Y} $ and $ H_{Y|X} $ and the maximizers of $ I_{X;Y} $, and hence could
also be called \emph{maximum mutual information couplings} for example.

From the last observation we see that minimum entropy couplings express the
largest dependence (measured by $ I_{X;Y} $) of random variables having
particular marginal distributions; this is further discussed in Section
\ref{renyiax}.

\begin{proposition}
\label{Hmin}
 The functional
$ H_{\min} : \Gamma_{n}^\sml{(1)} \times \Gamma_{m}^\sml{(1)} \to \mathbb{R} $,
defined by $ H_{\min}(P,Q) = \inf_{ S \in {\mathcal C}(P,Q) } H(S) $, is
continuous in $ (P,Q) $ (with respect to the product topology).
\end{proposition}
\begin{remark}
 Throughout the paper, the assumed topology on $ \Gamma_{n}^\sml{(1)} $ and
similar sets of probability distributions is the one induced by the $ \ell_1 $
norm, denoted $ \norm{\cdot} $.
\end{remark}
\begin{proof}
 The problem at hand is a constrained optimization problem, and we will use
a standard result in the field -- the Berge's maximum theorem
\cite[Thm 9.14]{sundaram}.
To see that the conditions of the theorem are satisfied, observe that entropy
is continuous over $ \Gamma_{n\times m}^\sml{(2)} $, and that the mapping
$ (P,Q) \mapsto {\mathcal C}(P,Q) $, viewed as a correspondence%
\footnote{\,The term correspondence denotes a set-valued map (i.e.,
multi-valued map).
Much of the study of such maps was motivated by their applications in
mathematical economics.
For a definition of continuity of correspondences, as well as the related
notions of lower and upper hemi-continuity, see \cite{sundaram}.},
is compact-valued and continuous \cite{bergin}.
\end{proof}

Berge's maximum theorem also implies that the mapping
$ (P,Q) \mapsto \arginf_{ S \in {\mathcal C}(P,Q) } H(S) $, which
maps distributions $ P, Q $ to the set of minimum entropy couplings
in $ {\mathcal C}(P,Q) $, is a compact-valued upper hemi-continuous
correspondence on $ \Gamma_{n}^\sml{(1)} \times \Gamma_{m}^\sml{(1)} $.
It is in fact finite-valued because minimum entropy couplings are
necessarily vertices of $ {\mathcal C}(P,Q) $, as commented above.
In the following we analyze the computational complexity of finding
an element of the set $ \arginf_{ S \in {\mathcal C}(P,Q) } H(S) $.

Let \textsc{Minimum entropy coupling} be the following computational
problem: Given probability distributions $ P = (p_1, \ldots, p_n) $ and
$ Q = (q_1, \ldots, q_m) $ (with%
\footnote{\,The probabilities being rational numbers is not just an algorithmic
requirement, it is also the most important case in statistics, where empirical
distributions and contingency tables have precisely such entries.}
$ p_i, q_j \in \mathbb{Q} $), find the minimum entropy coupling of $ P $
and $ Q $.
The proof of the following theorem relies on the well-known NP-complete
problem \cite{garey}
\displayproblem
  {Subset sum}
  {Positive integers $ d_1, \ldots, d_n $ and $ s $.}
  {Is there a $ J \subseteq \{ 1, \ldots, n \} $ such that $ \sum_{j \in J} d_j = s $ ?}

\begin{theorem}
\label{minentropy}
 \textsc{Minimum entropy coupling} is NP-hard. 
\end{theorem}
\begin{proof}
 We demonstrate a reduction from the \textsc{Subset sum} to the
\textsc{Minimum entropy coupling}.
Let there be given an instance of the \textsc{Subset sum}, i.e., a set of positive
integers $ s; d_1, \ldots, d_n $, $ n \geq 2 $.
Let $ D = \sum_{i=1}^{n} d_i $, and let $ p_i = d_i / D $, $ q = s / D $
(assume that $ s < D $, the problem otherwise being trivial).
Denote $ P = ( p_1, \ldots, p_n ) $ and $ Q = ( q, 1-q ) $.
The question we are trying to answer is whether there is a
$ J \subseteq \{ 1, \ldots, n \} $ such that $ \sum_{j \in J} d_j = s $, 
i.e., such that $ \sum_{j \in J} p_j = q $.
Observe that this happens if and only if there is a matrix $ S $ with row sums
$ P = ( p_1, \ldots, p_n) $ and column sums $ Q = ( q, 1-q ) $, which has exactly
one nonzero entry in every row (or, in probabilistic language, a distribution
$ S \in {\mathcal C}(P,Q) $ such that $ Y $ deterministically depends on $ X $).
We know that in this case, and only in this case, the entropy of $ S $ would be
equal to $ H(P) $ \cite{cover}, which is by \eqref{ineqH} a lower bound on entropy
over $ {\mathcal C}(P,Q) $.
In other words, if such a distribution exists, it must be the minimum entropy
coupling.
Therefore, if we could find the minimum entropy coupling, we could easily decide
whether it has one nonzero entry in every row, thereby solving the given instance
of the \textsc{Subset sum}.
\end{proof}

\begin{remark}
 We have shown that the problem of deciding whether there is a distribution
$ S \in {\mathcal C}(P,Q) $ with $ H(S) = H(P) $ is NP-complete even when
the distribution $ Q $ is allowed to have only two masses.
In this case it is equivalent to the \textsc{Subset sum} problem and represents
its information theoretic analogue (and it implies the hardness of the
\textsc{Minimum entropy coupling}).
When this restriction on $ Q $ is removed, the problem is equivalent to deciding
whether there exist subsets with prescribed sums $ s_1, \ldots, s_m $.
This problem is NP-complete in the strong sense \cite{garey} because it is
a generalization of the \textsc{3-Partition} problem which we recall below.
Since the reduction in the proof of the previous theorem is clearly
pseudo-polynomial \cite{garey} (it is just a division of all numbers by
$ D $), it follows that \textsc{Minimum entropy coupling} is strongly NP-hard.
\end{remark}

It would be interesting to determine whether the \textsc{Minimum entropy
coupling} belongs to FNP \cite{pap}, but this appears to be quite difficult.
Namely, given the optimal solution, it is not obvious how to verify (in polynomial
time) that it is indeed optimal.
A similar situation arises with the decision version of this problem:
Given $ P $ and $ Q $ and a threshold $ h $, is there a distribution
$ S \in {\mathcal C}(P,Q) $ with entropy $ H(S) \leq h $?
Whether this problem belongs to NP is another interesting question (which we will
not be able to answer here).
We will not go into these details further; we mention instead one closely related
problem which has been studied in the literature:
\displayproblem
  { Sqrt sum }
  { Positive integers $d_1,\ldots,d_n$, and $k$. }
  { Decide whether $\sum_{i=1}^n \sqrt{d_i} \leq k$ ? }
This problem, though ``conceptually simple" and bearing certain resemblance
with the above decision version of the entropy minimization problem, is not
known to be solvable in NP \cite{etessami} (it is solvable in PSPACE).
\subsection{Optimization over ${\mathcal C}(P,m)$}
\label{onemarginal}
\begin{definition}
 \emph{Optimal channel} with $ m $ outputs and input distribution $ P $ is
a bivariate distribution $ S^* \in {\mathcal C}(P,m) $ that maximizes the
mutual information functional, i.e., 
\begin{equation}
  I_{X;Y}(S^*) = \sup_{ S \in {\mathcal C}(P,m) } I_{X;Y}(S) .
\end{equation}
\end{definition}

Since $ H(X) $ is fixed, maximizing $ I_{X;Y} $ over $ {\mathcal C}(P,m) $
is equivalent to minimizing the conditional entropy $ H(X|Y) $, and is the
only interesting optimization problem over domains of this form.
Namely, the minimizer of $ H(X,Y) $ and $ H(Y|X) $ over ${\mathcal C}(P,m)$
is any joint distribution having at most one nonzero entry in each row (i.e.,
such that $ Y $ deterministically depends on $ X $), and the maximizer is
$ P \times U_m $, where $ U_m $ is the uniform distribution over
$ \{ 1, \ldots, m \} $.

As in the case of minimum entropy couplings, existence of optimal channels
for any $ P \in \Gamma_{n}^\sml{(1)} $ and $ m \in \mathbb{N} $ follows from
the continuity of $ I_{X;Y} $ and the compactness of $ {\mathcal C}(P,m) $.
They are in general not unique.
Since $ I_{X;Y} $ is convex when one marginal is fixed \cite[Thm 2.7.4]{cover},
we again have a convex maximization problem and conclude that optimal
channels are vertices of $ {\mathcal C}(P,m) $.

By the Berge's maximum theorem we have the following claim.
\begin{proposition}
\label{Imax}
 The functional
$ I_{\max} : \Gamma_{n}^\sml{(1)} \times \mathbb{N} \to \mathbb{R} $,
defined by $ I_{\max}(P,m) = \sup_{ S \in {\mathcal C}(P,m) } I_{X;Y}(S) $,
is continuous in $ P $.
The mapping $ P \mapsto \argsup_{ S \in {\mathcal C}(P,m) } I_{X;Y}(S) $ is a
compact-valued upper hemi-continuous correspondence on $ \Gamma_{n}^\sml{(1)} $.
\hfill \qed
\end{proposition}

To study the computational complexity of the above problem, define \textsc{Optimal
channel} as follows: Given a distribution $ P = ( p_1, \ldots, p_n ) $ and a number
$ m $ (with $ p_i \in \mathbb{Q} $, $ m \in \mathbb{N} $), find the distribution
$ S \in {\mathcal C}(P,m) $ which maximizes the mutual information.
This problem is the reverse of the channel capacity in the sense that now the input
distribution (the distribution of the source) is fixed, and the maximization is over
the conditional distributions.
In other words, given a source, we are asking for the channel with a given number
of outputs which has the largest mutual information.

We will use the well-known \textsc{Partition} (or \textsc{Number partitioning}) 
problem \cite{garey}. 
\displayproblem
  { Partition }
  { Positive integers $ d_1, \ldots, d_n $. }
  { Is there a partition of $ \{ d_1, \ldots, d_n \} $ into two subsets with equal sums? }
This is clearly a special case of the \textsc{Subset sum}.
It can be solved in pseudo-polynomial time by dynamic programming methods \cite{garey}.
But the following closely related problem is much harder.
\displayproblem
  { 3-Partition }
  { Positive integers $ d_1, \ldots, d_{3m} $ with $ s/4 < d_j < s/2 $, where
    $ s = \sum_{j} d_j / m $. }
  { Is there a partition of $ \{ 1, \ldots, 3m \} $ into $ m $ subsets $ J_1, \ldots, J_m $
    (disjoint and covering $ \{ 1, \ldots, 3m \} $) such that $ \sum_{ j \in J_i } d_j $
     are all equal? (The sums are necessarily $ s $ and every $ J_i $ has $ 3 $ elements.)}
This problem is NP-complete in the strong sense \cite{garey}, i.e., no pseudo-polynomial
time algorithm for it exists unless P=NP.
\begin{theorem}
\label{channel}
  \textsc{Optimal channel} is NP-hard. 
\end{theorem}
\begin{proof}
 We prove the claim by reducing 3-\textsc{Partition} to \textsc{Optimal channel}.
Let there be given an instance of the 3-\textsc{Partition} as described above, and
let $ p_i = d_i / D $, where $ D = \sum_i d_i $.
Observe that a partition with desired properties exists if and only if there is
a matrix $ C \in \mathcal{C}(P,m) $, $ P = (p_i) $, having column sums $ s/D = 1/m $
and having exactly one nonzero entry in every row (i.e., if and only if a bivariate
distribution $ C \in \mathcal{C}(P,m) $ exists with the uniform second marginal
$ \left( \frac{1}{m}, \ldots, \frac{1}{m} \right) $, and such that $ Y $
deterministically depends on $ X $).
Furthermore, a distribution $ C \in \mathcal{C}(P,m) $ has such properties if
and only if it satisfies $ I_{X;Y}(C) = \log m $ (to see this, observe that
\begin{inparaenum}[(i)]
\item
$ I_{X;Y}(S) \leq H(S_Y) $, where $ S_Y $ is the second marginal of $ S $, with
equality if and only if $ S $ has at most one nonzero entry in every row, and
\item
$ H(S_Y) \leq \log m $, whenever $ S \in \mathcal{C}(P,m) $, with equality if and
only if $ S_Y $ is uniform).
\end{inparaenum}
Since $ \log m $ is an upper bound on $ I_{X;Y} $ over $ \mathcal{C}(P,m) $,
such a distribution $ C $ would necessarily be the maximizer of $ I_{X;Y} $.
To conclude, if we could solve the \textsc{Optimal channel} with instance
$ ( p_1, \ldots, p_{3m} ); m $, we could easily decide whether the maximizer has
column sums $ 1/m $ and exactly one nonzero entry in every row, thereby solving
the original instance of the 3-\textsc{Partition}.\phantom{xxxxx}
\end{proof}

Note that the problem remains NP-hard even when the number of channel outputs ($ m $)
is fixed in advance and is not a part of the input instance.
For example, maximization of $ I_{X;Y} $ over $ {\mathcal C}(P,2) $ is essentially
equivalent to the \textsc{Partition} problem.
Furthermore, since the transformation in the proof of Theorem \ref{channel} is
pseudo-polynomial \cite{garey}, \textsc{Optimal channel} is strongly NP-hard and,
unless P=NP, has no pseudo-polynomial time algorithm.

\subsection{Comments and generalizations}
\label{generalizations}

In this subsection we put the above optimization problems in a more general context,
and discuss their relevance and certain generalizations.

\subsubsection{Entropy minimization}

Entropy minimization, taken in the broadest sense, is a very important problem.
Watanabe \cite{pattern} has shown, for example, that many algorithms for clustering
and pattern recognition can be characterized as suitably defined entropy minimization
problems.
In theoretical computer science, a class of combinatorial optimization problems based
on entropy minimization has been studied extensively (see \cite{cardinal} and the
references therein).
These include minimum entropy set cover, minimum entropy graph coloring, minimum
entropy orientation, etc.

A much more familiar problem in information theory is that of entropy maximization.
The so-called \emph{Maximum entropy principle} formulated by Jaynes \cite{jaynes}
states that, among all proba\-bility distributions satisfying certain constraints
(expressing our knowledge about the system), one should pick the one with maximum
entropy.
It has been recognized by Jaynes, as well as many other researchers, that this
choice gives the least biased, the most objective distribution consistent with
the information one possesses about the system.
Consequently, the problem of maximizing entropy under constraints has been thoroughly
studied (see, e.g., \cite{harr,kapur1}).
It has been argued, however, that minimum entropy distributions can also be of
interest in many contexts.
The MinMax information measure, for example, has been introduced \cite{kapur2,yuan}
as a measure of the amount of information contained in a given set of constraints,
and it is based both on maximum and minimum entropy distributions.

One could formalize the problem of entropy minimization as follows:
Given a polytope (by a system of inequalities with rational coefficients, say) in
the set of probability distributions, find the distribution $ S^* $ which minimizes
the entropy functional $ H $.
(If the coefficients are rational, then all the vertices are rational, i.e., have
rational coordinates.
Therefore, the minimum entropy distribution has finite description and is well-defined
as an output of a computational problem.)
This problem is strongly NP-hard and remains such over transportation polytopes,
as established above.

\subsubsection{R\'enyi entropy minimization}

R\'enyi entropy of order $ \alpha \geq 0 $ of a random variable $ X $ with
distribution $ P $ is defined as
\begin{equation}
  H_\alpha(X) \equiv H_\alpha(P) =
     \frac{ 1 }{ 1 - \alpha } \log \mysum_{i} p_i^\alpha ,
\end{equation}
with
\begin{equation}
  H_0(P) = \lim_{ \alpha \to 0 } H_\alpha(P) = \log |P| ,
\end{equation}
and
\begin{equation}
  H_1(P) = \lim_{ \alpha \to 1^+ } H_\alpha(P) = H(P) .
\end{equation}
It was introduced by R\'enyi \cite{renyi} on axiomatic grounds as a generalization
of the Shannon entropy, and represents an important functional in information theory.
Joint R\'enyi entropy of the pair $ (X, Y) $ having distribution $ S = (s_{i,j}) $
is naturally defined as
\begin{equation}
  H_\alpha(X,Y) \equiv H_\alpha(S) =
     \frac{ 1 }{ 1 - \alpha } \log \mysum_{i,j} s_{i,j}^\alpha .
\end{equation}
Due to subadditivity (for $ \alpha < 1 $) and superadditivity (for $ \alpha > 1 $)
properties of the function $ x^\alpha $ for $ x \geq 0 $, it follows that
\begin{equation}
  H_\alpha(X,Y) \geq \max \big\{ H_\alpha(X) , H_\alpha(Y) \big\}
\end{equation}
with equality if and only if $ X $ is a function of $ Y $, or vice versa.
In the same way as in Theorem \ref{minentropy} we then conclude that the problem of
minimization of the R\'enyi entropies $ H_\alpha $ over arbitrary polytopes is strongly
NP-hard, for any $ \alpha \geq 0 $.
Note that, for $ \alpha > 1 $, this problem is equivalent to the maximization of the
$ \ell_\alpha $ norm (see also \cite{maxnorm,maxnorm2} for different proofs of the
NP-hardness of norm maximization).
Interestingly, however, the minimization of the R\'enyi entropy of order $ \infty $,
defined as
\begin{equation}
  H_\infty(P) = \lim_{ \alpha \to \infty } H_\alpha(P) = -\log \max_i p_i ,
\end{equation}
is polynomial-time solvable; it is equivalent to the maximization of the $ \ell_\infty $
norm \cite{maxnorm}.
For $ \alpha < 1 $, the minimization of R\'enyi entropy is equivalent to the
minimization of $ \ell_\alpha $ (which is not a norm in the strict sense), a problem
arising in compressed sensing \cite{minnorm}.

Hence, as we have seen throughout this section, various problems from computational
complexity theory can be reformulated as information-theoretic optimization problems.
(Observe also the similarity of the \textsc{Sqrt sum} and the minimization of R\'enyi
entropy of order $1/2$.)

\subsubsection{Other information measures}

Maximization of mutual information is also a problem of great importance in information
theory.
The so-called Maximum mutual information criterion has found many applications, e.g.,
for feature selection \cite{battiti} and the design of classifiers \cite{deng}.
Another familiar example is that of the capacity of a communication channel which is
defined precisely as the maximum of the mutual information between the input and the
output of a channel.

We have illustrated the general intractability of the problem of maximization
of $ I_{X;Y} $ by exhibiting two simple classes of polytopes over which the problem
is strongly NP-hard.
We also mention here one possible generalization of this problem -- maximization of
information divergence.
Namely, since for $ S \in {\mathcal C}(P,Q) $
\begin{equation}
  I_{X;Y}(S) = D(S||P\times Q) ,
\end{equation}
one can naturally consider the more general problem of maximizing $ D(S||T) $ when
$ S $ belongs to some convex region and $ T $ is fixed.
Related problems of finding maximizers of information divergence from exponential
families have been studied in \cite{ay, matus, rauh, rauh_phd}.

Formally, let \textsc{Information divergence maximization} be the following
computational problem:
Given a rational convex polytope $ \mathcal S $ in the set of probability distributions,
and a distribution $ T $, find the distribution $ S \in {\mathcal S} $ which maximizes
$ D(\cdot||T) $.
This is again a convex maximization problem because $ D(S||T) $ is convex in the pair
$ (S, T) $ \cite{csiszar}.

\begin{corollary}
  \textsc{Information divergence maximization} is NP-hard.
\hfill \qed
\end{corollary}

Note that the reverse problem, namely the minimization of information divergence,
defines an information projection of $ T $ onto the region $ \mathcal S $ \cite{csiszar}.

\subsubsection{Measures of statistical dependence}
\label{renyiax}

We conclude this subsection with one more generalization of the problem of maximization 
of mutual information.
Namely, this problem can also be seen as a statistical problem of expressing the largest
possible dependence between two given random variables.

Consider the following statistical scenario.
A system is described by two random variables (taking values in $ \mathbb{N} $) whose
joint distribution is unknown; only some constraints that it must obey are given.
The set of all distributions satisfying these constraints is usually called a statistical
model.
\begin{example}
 Suppose we have two correlated information sources obtained by independent
drawings from a discrete bivariate probability distribution, and suppose we
only have access to individual streams of symbols (i.e., streams of symbols
from either one of the sources, but not from both simultaneously) and can
observe the relative frequencies of the symbols in each of the streams.
We therefore ``know" probability distributions of both sources (say $ P $ and $ Q $),
but we don't know how correlated they are. Then the ``model" for this joint
source would be $ {\mathcal C}(P,Q) $. In the absence of any additional
information, we must assume that some $ S \in {\mathcal C}(P,Q) $ is the ``true"
distribution of the source.
\end{example}

Given such a model, we may ask the following question: What is the largest
possible dependence of the two random variables? How correlated can they possibly
be? This question can be made precise once a dependence measure is specified, and
this is done next.

A. R\' enyi \cite{renyidep} has formalized the notion of probabilistic dependence
by presenting axioms which a ``good" dependence measure $ \rho $ should satisfy.
These axioms, adapted for discrete random variables, are listed below.
\begin{enumerate}%
 \item[(A)]  $\rho(X,Y)$ is defined for any two random variables $ X $, $ Y $, neither 
             of which is constant with probability $1$.%
 \item[(B)]  $0 \leq \rho(X,Y) \leq 1$.%
 \item[(C)]  $\rho(X,Y) = \rho(Y,X)$.%
 \item[(D)]  $\rho(X,Y) = 0$ iff $X$ and $Y$ are independent.%
 \item[(E)]  $\rho(X,Y) = 1$ iff $X=f(Y)$ or $Y=g(X)$.%
 \item[(F)]  If $f$ and $g$ are injective functions, then $\rho(f(X),g(Y)) = \rho(X,Y)$.%
\end{enumerate}%
Actually, R\' enyi considered axiom (E) to be too restrictive and demanded only the
``if part". It has been argued subsequently \cite{bell}, however, that this is a
substantial weakening.
We will find it convenient to consider the stronger axiom given above.
As an example of a good measure of dependence, one could take precisely the mutual
information; its normalized variant $ I(X;Y) / \min\{ H(X), H(Y) \} $ satisfies
all the above axioms.

Let us now formalize the question asked above.
Namely, let \textsc{maximal $\rho$--dependence} be the following problem:
Given two probability distributions $ P = ( p_1, \ldots, p_n ) $ and
$ Q = ( q_1, \ldots, q_m ) $, $ p_i, q_j \in \mathbb{Q} $, find the distribution
$ S \in {\mathcal C}(P,Q) $ which maximizes $ \rho $.
The proof of the following claim is identical to the one given for mutual information
(entropy) in Section \ref{bothmarginals} and is therefore omitted.
\begin{theorem}
 Let $ \rho $ be a measure of dependence satisfying axioms \emph{(A)--(F)}.
Then \textsc{maximal $ \rho $--dependence} is NP-hard.
\hfill \qed
\end{theorem}

The intractability of the problem over more general statistical models is now
a simple consequence.

\subsection{Infinite alphabets}
\label{infinite}

We conclude this section with a discussion on the properties of information
measures over domains of the form $ {\mathcal C}(P,Q) $ and $ {\mathcal C}(P,m) $
in the case when the distributions $ P $ and $ Q $ have possibly infinite supports.
The notation is similar to the finite alphabet case, for example
\begin{equation}
 \begin{aligned}
  \Gamma^\sml{(1)} &= \Bigg\{ (p_i)_{i\in\mathbb{N}}\,:\,p_i\geq0\,,\,\mysum_{i} p_i = 1 \Bigg\},   \\
  \Gamma^\sml{(2)} &= \Bigg\{ (p_{i,j})_{i,j\in\mathbb{N}}\,:\,p_{i,j}\geq0\,,\,\mysum_{i,j} p_{i,j} = 1 \Bigg\}. 
 \end{aligned}
\end{equation}

Let also
$ \ell_1^\sml{(2)} = \big\{ (x_{i,j})_{ i,j \in \mathbb{N} } : \sum_{i,j} | x_{i,j} | < \infty \big\} $.
This is the familiar $ \ell_1 $ space, only defined for two-dimensional sequences.
It clearly shares all the essential properties of $ \ell_1 $, completeness being the
one we will exploit.
The metric understood is
\begin{equation}
  \norm{x-y} = \mysum_{i,j} |x_{i,j} - y_{i,j}|, 
\end{equation}
for $ x, y \in \ell_1^\sml{(2)} $.
In the context of probability distributions, this distance is usually called the total
variation distance (actually, it is twice the total variation distance, see \eqref{l1variation}).

\begin{proposition}
\label{comp}
 For any $ P, Q \in \Gamma^\sml{(1)} $ and $ m \in \mathbb{N} $, $ {\mathcal C}(P,Q) $
and $ {\mathcal C}(P,m) $ are compact.
\end{proposition}
\begin{proof}
 A metric space is compact if and only if it is complete and totally bounded \cite{bredon};
these facts are demonstrated below.
\end{proof}

\begin{lemma}
\label{complete}
  $ {\mathcal C}(P,Q) $ and $ {\mathcal C}(P,m) $ are complete metric spaces.
\end{lemma}
\begin{proof}
 It is enough to show that $ {\mathcal C}(P,Q) $ and $ {\mathcal C}(P,m) $ are closed
in $\ell_1^\sml{(2)}$ because closed subsets of complete spaces are always complete.
In other words, it suffices to show that for any sequence $ S_n \in {\mathcal C}(P,Q) $
converging to some $ S \in \ell_1^\sml{(2)} $ (in the sense that $ \norm{S_n - S} \to 0 $),
we have $ S \in {\mathcal C}(P,Q) $.
This is straightforward.
If $ S_n $ all have the same marginals ($ P $ and $ Q $), then $ S $ must also have
these marginals, for otherwise the distance between $ S_n $ and $ S $ would be lower
bounded by the distance between the corresponding marginals
\begin{equation}
  \mysum_{i,j} \left|S(i,j) - S_n(i,j)\right| \geq \mysum_{i} \bigg| \mysum_{j} ( S(i,j) - S_n(i,j) ) \bigg| 
\end{equation}
and hence could not decrease to zero.
The case of $ {\mathcal C}(P,m) $ is similar.
\end{proof}

For our next claim, recall that a set $ E $ is said to be totally bounded if it
has a finite covering by $ \epsilon $-balls, for any $ \epsilon > 0 $.
In other words, for any $ \epsilon > 0 $, there exist $ x_1, \ldots, x_K \in E $
such that $ E \subseteq \bigcup_k {\mathcal B}(x_k, \epsilon) $, where
$ {\mathcal B}(x_k, \epsilon) $ denotes the open ball around $ x_k $ of radius $ \epsilon $.
The points $ x_1, \ldots, x_K $ are then called an $ \epsilon $-net for $ E $.

\begin{lemma}
\label{bounded}
  $ {\mathcal C}(P,Q) $ and $ {\mathcal C}(P,m) $ are totally bounded.
\end{lemma}
\begin{proof}
 We prove the statement for $ {\mathcal C}(P,Q) $, the proof for $ {\mathcal C}(P,m) $ 
is very similar.
Let $ P $, $ Q $, and $ \epsilon > 0 $ be given.
We need to show that there exist distributions $ S_1, \ldots, S_K \in {\mathcal C}(P,Q) $
such that $ {\mathcal C}(P,Q) \subseteq \bigcup_{k} {\mathcal B}(S_k, \epsilon) $, and
this is done in the following.
There exists $N$ such that $ \sum_{ i = N + 1 }^{\infty} p_i < \frac{ \epsilon }{ 6 }$
and $ \sum_{ j = N + 1 }^{\infty} q_j < \frac{ \epsilon }{ 6 }$.
Observe the truncations of the distributions $ P $ and $ Q $, namely $ ( p_1, \ldots, p_N ) $
and $ ( q_1, \ldots, q_N ) $.
Assume that $ \sum_{i=1}^N p_i \geq \sum_{j=1}^N q_j $, and let
$ r = \sum_{i=1}^N p_i - \sum_{j=1}^N q_j = \sum_{j=N+1}^\infty q_j - \sum_{i=N+1}^\infty p_i $
(otherwise, just interchange $ P $ and $ Q $).
Now let $ P^{(N)} = ( p_1, \ldots, p_N ) $ and $ Q^{(N,r)} = ( q_1, \ldots, q_N, r ) $,
and observe $ {\mathcal C}( P^{(N)}, Q^{(N,r)} ) $.
(Adding $ r $ was necessary for $ {\mathcal C}( P^{(N)}, Q^{(N,r)} ) $ to be nonempty.)
This set is closed (see the proof of Lemma \ref{complete}) and bounded in
$ \mathbb{R}^{ N \times (N+1) } $, and hence it is compact by the Heine-Borel theorem.
This further implies that it is totally bounded and has an $ \frac{ \epsilon }{ 6 } $-net,
i.e., there exist $ T_1, \ldots, T_K \in {\mathcal C}( P^{(N)}, Q^{(N,r)} ) $ such that
$ {\mathcal C}( P^{(N)}, Q^{(N,r)} ) \subseteq \bigcup_{k} {\mathcal B}( T_k, \frac{ \epsilon }{ 6 } ) $.
Now construct distributions $ S_1, \ldots, S_K \in {\mathcal C}(P, Q) $ by ``padding"
$ T_1, \ldots, T_K $.
Namely, take $ S_k $ to be any distribution in $ {\mathcal C}(P,Q) $ which coincides with
$ T_k $ on the first $ N \times N $ coordinates, for example
\begin{equation}
  S_k(i,j) = \begin{cases}
               T_k(i,j) ,  &  i,j \leq N  \\ 
               0 ,         &  i > N, j\leq N  \\ 
               T_k(i,N+1)\cdot{q_j}/{\sum_{l=N+1}^\infty q_l} ,         &  i\leq N, j > N  \\ 
               p_i\cdot{q_j}/{\sum_{l=N+1}^\infty q_l} ,         &  i,j > N .
             \end{cases}
\end{equation}
Understanding that $ T_k(i,j) = 0 $ for $ i > N $ or $ j > N + 1 $, we have
\begin{equation}
\begin{aligned}
  \norm{ T_k - S_k } &=     \sum_{i=1}^\infty \sum_{j=N+2}^\infty S_k(i,j) + \sum_{i=1}^\infty | T_k(i,N+1) - S_k(i,N+1) |  \\
                     &\leq  \sum_{i=1}^\infty \sum_{j=N+1}^\infty S_k(i,j) + \sum_{i=1}^\infty T_k(i,N+1)  \\
                     &=     \sum_{j=N+1}^\infty q_j + r  <  \frac{ \epsilon }{ 6 } + \frac{ \epsilon }{ 6 }  =  \frac{ \epsilon }{ 3 } .
\end{aligned}
\end{equation}
We prove below that $ S_k $'s are the desired $ \epsilon $-net for $ {\mathcal C}(P,Q) $,
i.e., that any distribution $ S \in {\mathcal C}(P,Q) $ is at distance at most $ \epsilon $
from some $ S_\ell $, $ \ell \in \{ 1, \ldots, K \} $ ($ \norm{ S-S_\ell } < \epsilon $).
Observe some $ S \in {\mathcal C}(P,Q) $, and let $ S' $ be its $ N \times N $ truncation
\begin{equation}
  S'(i,j) = \begin{cases}
               S(i,j) ,  &  i,j\leq N  \\ 
               0 ,       &  \text{otherwise}.
             \end{cases}
\end{equation}
Note that $ S' $ is not a distribution, but that does not affect the proof.
Note also that the marginals of $ S' $ are bounded from above by the marginals of $ S $,
namely $ q_j' = \sum_i S'(i,j) \leq q_j $ and $ p_i' = \sum_j S'(i,j) \leq p_i $.
Finally, we have $ \norm{ S - S' } < \frac{ \epsilon }{ 3 }$ because the total mass of
$ S $ on the coordinates where $ i > N $ or $ j > N $ is at most $ \frac{ \epsilon }{ 3 } $.
The next step is to create $ S'' \in {\mathcal C}( P^{(N)}, Q^{(N,r)} ) $ by adding masses
to $ S' $ on the $ N \times (N + 1) $ rectangle.
One way to do this is as follows.
Let
\begin{align}
  u_i &= \begin{cases}
               p_i - p_i' ,  &  i \leq N  \\ 
               0 ,           &  i > N 
        \end{cases} ,
  \\
  v_j &= \begin{cases}
               q_j - q_j' ,  &  j \leq N  \\ 
               r ,           &  j = N + 1  \\
               0 ,           &  j > N + 1 
        \end{cases} ,
\end{align}
and let $ U = (u_i) $, and $ V = (v_j) $, and $ c = \sum_i u_i = \sum_j v_j $ (to see that
these two sums are equal write
$ \sum_{i} u_i - \sum_{j} v_j = \sum_{i = 1}^N (p_i - p_i') - \sum_{j = 1}^N (q_j - q_j') - r $
which is equal to zero by the definition of $ r $ and due to the fact that
$ \sum_{i = 1}^N p_i' = \sum_{j = 1}^N q_j' = \sum_{i=1}^N \sum_{j = 1}^N S(i,j) $).
Now define $ S'' $ by
\begin{equation}
  S'' = S' + \frac{1}{c}U\times V . 
\end{equation}
It is easy to verify that $ S'' \in {\mathcal C}( P^{(N)}, Q^{(N,r)} ) $ and that
$ \norm{ S'-S'' } < \frac{ \epsilon }{ 6 } $ because the total mass added is
\begin{equation}
 \begin{aligned}
  c = \mysum_{i=1}^{N} ( p_i - p_i' ) &= \mysum_{i=1}^{N} \mysum_{j=1}^{\infty} ( S(i,j) -  S'(i,j) )   \\
                                      &= \mysum_{i=1}^{N} \mysum_{j=N+1}^{\infty} S(i,j)  \\
                                      &\leq \mysum_{j=N+1}^{\infty} q_j < \frac{\epsilon}{6} . 
 \end{aligned}
\end{equation}
Now recall that $ T_k $'s form an $ \frac{ \epsilon }{ 6 } $-net for
$ {\mathcal C}( P^{(N)}, Q^{(N,r)} ) $ and consequently that there exists some $ T_\ell $,
$ \ell \in \{ 1, \ldots, K \} $, with $ \norm{ S''-T_\ell } < \frac{ \epsilon }{ 6 } $.
To put this all together, write
\begin{equation}
  \norm{S-S_\ell} \leq \norm{S-S'} + \norm{S'-S''} +  
                       \norm{S''-T_\ell} + \norm{T_\ell-S_\ell} < \epsilon , 
\end{equation}
which completes the proof.
\end{proof}

The following claim shows that imposing certain restrictions on the marginal
distributions ensures the continuity of Shannon information measures and
existence of their extrema. In contrast, without any restrictions, these
functionals are known to be discontinuous at every point of $ \Gamma^\sml{(2)} $.
(Entropy is, however, sequentially continuous at any power bounded distribution
in the topology of information divergence \cite[Thm 21]{harr2}; this weaker
notion of continuity is useful for many applications in probability theory.)

\begin{theorem}
\label{cont}
 Let $ P, Q \in \Gamma^\sml{(1)} $ and $ m \in \mathbb{N} $, and assume that
$ Q $ has finite entropy. Then Shannon information measures are uniformly
continuous and attain their extrema over $ {\mathcal C}(P,Q) $ and
$ {\mathcal C}(P,m) $.
\end{theorem}
\begin{proof}
 Continuity over $ {\mathcal C}(P,Q) $ and $ {\mathcal C}(P,m) $ is a
special case of \cite[Thm 4.3]{harr} and can thus be established
by exhibiting \emph{cost-stable codes} for these statistical models.
We also give here a more direct proof (which can be extended to prove
Theorem \ref{allshannon}).
Write
\begin{equation}
\label{mutinf}
   H_Y(S) = I_{X;Y}(S) + H_{Y|X}(S) . 
\end{equation}
The functional
$ H_{Y|X}(S) = \mysum_{i,j} s_{i,j} \log \frac{ p_i }{ s_{i,j} } $ is
lower semi-continuous because it is a sum of nonnegative continuous
functions. The functional $ I_{X;Y} $ is also lower semi-continuous since
\begin{equation}
\label{KL}
  I_{X;Y}(S) = D(S||P\times Q), 
\end{equation}
and information divergence $ D(S||T) $ is known to be jointly lower
semi-continuous in the distributions $ S $ and $ T $ \cite[Thm 3.1]{topsoe2}.
But since the sum of these two functionals is a constant
$ H_Y(S) = H(Q) < \infty $, both of them must be continuous. The
continuity of $ H_{X|Y} $ and $ H_{X,Y} $ follows from \eqref{identity}. 

Now consider $ {\mathcal C}(P,m) $. In \cite{ho} it is shown that
$ H(Y|X) $ and $ I(X;Y) $ are continuous when the alphabet of $ Y $
is finite and fixed, which is what we have here. And since $ H(X) = H(P) $
is fixed, $ H(X|Y) $ and $ H(X,Y) $ are also continuous (if $ H(P) = \infty $
then they are infinite over the entire $ {\mathcal C}(P,m) $, but we
also take this to mean that they are continuous).

Uniform continuity and the fact that the above functionals attain their
extrema over $ {\mathcal C}(P,Q) $ and $ {\mathcal C}(P,m) $ now follow
from the compactness of these domains.
\end{proof}

Regarding the extrema of information measures, we note that Proposition
\ref{Hmin} fails in the case of unbounded alphabets (when
$ (P,Q) \in \Gamma^\sml{(1)} \times \Gamma^\sml{(1)} $). Namely, the
functional $ H_{\min}(P,Q) $ is discontinuous at every $ (P,Q) $ with
$ H(P), H(Q) < \infty $. This follows easily from the discontinuity of
entropy. However, Proposition \ref{Imax} remains valid because $ I_{X;Y} $
is continuous when one of the alphabets is finite \cite{ho}.

The argument in the proof of Theorem \ref{cont} can easily be adapted to
prove the following more general claim which gives necessary and sufficient
conditions for the convergence of entropy in terms of other information
measures.
\begin{theorem}
 \label{allshannon}
  Let $ S \in \Gamma^\sml{(2)} $ be a bivariate probability distribution with
  finite entropy, $ H_{X,Y}(S) < \infty $. Then the following statements
  are equivalent:
  \begin{enumerate}
   \item $ H_{X,Y} $ is continuous at $ S $,
   \item $ H_X $ and $ H_Y $ are continuous at $ S $,
   \item $ I_{X;Y} $, $ H_{X|Y} $, and $ H_{Y|X} $ are continuous at $ S $.
  \end{enumerate}
\end{theorem}
\begin{proof}
  Note first that when $ S_n \to S $, then also $ P_n \to P $, $ Q_n \to Q $,
  and $ P_n \times Q_n \to P \times Q $, where $ P_n, Q_n $, and $ P, Q $
  are the marginals of $ S_n $ and $ S $, respectively. Now all implications
  follow from \eqref{identity} and the fact that the functionals in
  question are lower semi-continuous.
\end{proof}

\section{Metrics from couplings}
\label{distance}

Apart from many of their other uses, couplings are very convenient for defining
metrics on the space of probability distributions.
There are many interesting metrics defined via so-called ``optimal" couplings.
We illustrate this point below using one familiar example, and then define new
information-theoretic metrics based on the minimum entropy coupling.
Similar approaches are also used in the literature for defining measures of
distortion (that are not necessarily metrics) between random objects; see, e.g.,
\cite{gray} for the corresponding definitions and their applications in rate
distortion theory.

Given two probability distributions $ P $ and $ Q $, one could measure the
``distance" between them as follows.
Consider all possible random pairs $\left(X,Y\right)$ with marginal distributions
$ P $ and $ Q $.
Then define some measure of dissimilarity of $ X $ and $ Y $, for example
$ {\mathbb{P}}( X \neq Y ) $, and minimize it over all such couplings (minimization
is necessary for the triangle inequality to hold).
Indeed, this example yields the well-known total variation distance \cite{peres}
\begin{equation}
\label{variation}
  d_\textsc{tv}(P,Q) = \inf_{{\mathcal{C}}(P,Q)} {\mathbb{P}}(X\neq Y) ,
\end{equation}
where the infimum is taken over all joint distributions of the random vector $ (X, Y) $
with marginals $ P $ and $ Q $.
Notice that a minimizing distribution (called a \emph{maximal coupling}, see, e.g.,
\cite{sason}) in \eqref{variation} is ``easy" to find because $ {\mathbb{P}}( X \neq Y ) $
is a linear functional in the joint distribution of $ (X, Y) $.
For the same reason, $ d_\textsc{tv}(P, Q) $ is easy to compute, but this is also
clear from the identity \cite{peres}
\begin{equation}
\label{l1variation}
  d_\textsc{tv}(P,Q) = \frac{1}{2}\mysum_{i} |p_i-q_i|. 
\end{equation}
We next define information-theoretic distances in a similar manner.

\subsection{Entropy metrics}

Let $(X,Y)$ be a random pair with joint distribution $S$ and marginal distributions 
$P$ and $Q$. The total information contained in these random variables is $H(X,Y)$, 
while the information contained simultaneously in both of them (or the information
they contain about each other) is measured by $I(X;Y)$. One is then tempted to take
as a measure of their dissimilarity\footnote{\,Drawing a familiar information-theoretic 
Venn diagram \cite{cover} makes it clear that this is a measure of ``dissimilarity" 
of two random variables.}
\begin{equation}
 \label{delta}
  \Delta_1(X,Y) \equiv \Delta_1(S) = H(X,Y) - I(X;Y)    
                                   = H(X|Y) + H(Y|X).
\end{equation}
Indeed, this quantity (introduced by Shannon \cite{shannon}, and usually referred to 
as the \emph{entropy metric} \cite{csiszar}) satisfies the properties of a pseudometric 
\cite{csiszar}. In a similar way one can show that the following is also a pseudometric
\begin{equation}
 \label{deltamaks}
  \Delta_\infty(X,Y) \equiv \Delta_\infty(S) = \max\big\{H(X|Y),H(Y|X)\big\}, 
\end{equation}
as are the normalized variants of $\Delta_1$ and $\Delta_\infty$ \cite{metrics}. 
These pseudometrics have found numerous applications (see for example \cite{yao}) 
and have also been considered in an algorithmic setting \cite{kolmo}. 

One can further generalize these definitions to obtain a family of pseudometrics. 
This generalization is akin to the familiar $\ell_p$ distances. Let 
\begin{equation}
 \label{deltap}
  \Delta_p(X,Y) \equiv \Delta_p(S) = \big( H(X|Y)^p + H(Y|X)^p \big)^\frac{1}{p} , 
\end{equation}
for $p\geq1$. Observe that $ \lim_{p\to\infty} \Delta_p(X,Y) = \Delta_\infty(X,Y) $,
justifying the notation. 
\begin{proposition}
  $\Delta_p(X,Y)$ satisfies the properties of a pseudometric, for all $ p \in [1, \infty] $.
\end{proposition}
\begin{proof}
  Nonnegativity and symmetry are clear, as is the fact that $ \Delta_p(X,Y) = 0 $ if
(but not only if) $ X = Y $ with probability one.
The triangle inequality remains.
Following the proof for $ \Delta_1 $ from \cite[Lemma 3.7]{csiszar}, we first observe
that $ H(X|Y) \leq H(X|Z) + H(Z|Y) $, wherefrom
\begin{equation}
    \Delta_p(X,Y) \leq \Big( \big( H(X|Z) + H(Z|Y) \big)^p +  
                             \big( H(Y|Z) + H(Z|X) \big)^p \Big)^\frac{1}{p} . 
\end{equation}
Now apply the Minkowski inequality ($ \pnorm{a+b} \leq \pnorm{a} + \pnorm{b} $) to
the vectors $ a = ( H(X|Z), H(Z|X) ) $ and $ b = ( H(Z|Y), H(Y|Z) ) $ to get
\begin{equation}
  \Delta_p(X,Y) \leq \Delta_p(X,Z) + \Delta_p(Z,Y) , 
\end{equation}
which was to be shown.
\end{proof}
\begin{remark}
 $ \Delta_p $ are pseudometrics on the space of random variables over the same 
probability space.
Namely, for $ \Delta_p $ to be defined, the joint distribution of $ (X, Y) $ must
be given because joint entropy and mutual information are not defined otherwise.
Equation \eqref{deltainf} below defines the distance between random variables (more 
precisely, between their distributions) that does not depend on the joint distribution. 
\end{remark}

Having defined measures of dissimilarity, we can now define the corresponding 
distances
\begin{equation}
 \label{deltainf}
  \underline{\Delta}_p(P,Q) = \inf_{S\in{\mathcal{C}}(P,Q)} \Delta_p(S) . 
\end{equation}
The case $p=1$ has also been analyzed in some detail in \cite{vidyasagar}, motivated by 
the problem of optimal order reduction for stochastic processes. 

\begin{proposition}
 $\underline{\Delta}_p$ is a pseudometric on $\Gamma^\sml{(1)}$, for any $p\in[1,\infty]$. 
\end{proposition}
\begin{proof}
 Since ${\Delta}_p$ satisfies the properties of a pseudometric, we only need 
to show that these properties are preserved under the infimum. 
Nonnegativity and symmetry are clearly preserved.
Also, if $ P = Q $ then $ \underline{\Delta}_p(P,Q) = 0 $.
This is because $ S = \text{diag}(P) $ (distribution with masses $ p_i = q_i $ on
the diagonal and zeros elsewhere) belongs to $ {\mathcal{C}}(P,Q) $ in this case,
and for this distribution we have $ H_{X|Y}(S) = H_{Y|X}(S) = 0 $.
The triangle inequality is left. Let $X$, $Y$ and $Z$ be random variables with
distributions $ P $, $ Q $ and $ R $, respectively, and let their joint distribution
be specified.
We know that $ \Delta_p(X,Y) \leq \Delta_p(X,Z) + \Delta_p(Z,Y) $, and we have to prove
that
\begin{equation}
  \inf_{{\mathcal C}(P,Q)} \Delta_p(X,Y) \leq  
                 \inf_{{\mathcal C}(P,R)} \Delta_p(X,Z) + 
                 \inf_{{\mathcal C}(R,Q)} \Delta_p(Z,Y). 
\end{equation}
Since, from the above, 
\begin{equation}
  \inf_{{\mathcal C}(P,Q)} \Delta_p(X,Y) = 
               \inf_{{\mathcal C}(P,Q,R)} \Delta_p(X,Y)   
                \leq  \inf_{{\mathcal C}(P,Q,R)} \big\{\Delta_p(X,Z) + \Delta_p(Z,Y)\big\} 
\end{equation}
it suffices to show that 
\begin{equation}
\label{triangle}
   \inf_{{\mathcal C}(P,Q,R)} \big\{\Delta_p(X,Z) + \Delta_p(Z,Y)\big\} =   
     \inf_{{\mathcal C}(P,R)} \Delta_p(X,Z) + \inf_{{\mathcal C}(R,Q)} \Delta_p(Z,Y). 
\end{equation}
($ {\mathcal C}(P,Q,R) $ denotes the set of all three-dimensional distributions with
one-dimensional marginals $ P $, $ Q $, and $ R $, as the notation suggests.)
Let $ T \in {\mathcal C}(P,R) $ and $ U \in {\mathcal C}(R,Q) $ be the optimizing
distributions on the right-hand side (rhs) of \eqref{triangle}.
Observe that there must exist a joint distribution $ W \in {\mathcal C}(P,Q,R) $
consistent with $ T $ and $ U $ (for example, take $ w_{i,j,k} = t_{i,k} u_{k,j} / r_k $).
Since the optimal value of the lhs is less than or equal to the value at $ W $,
we have shown that the lhs of \eqref{triangle} is less than or equal to the rhs.
For the opposite inequality observe that the optimizing distribution on the lhs
of \eqref{triangle} defines some two-dimensional marginals $ T \in {\mathcal C}(P,R) $
and $ U \in {\mathcal C}(R,Q) $, and the optimal value of the rhs must be less than
or equal to its value at $ (T, U) $.
\end{proof}
\begin{remark}
 If $ \underline{\Delta}_p(P,Q) = 0 $, then $ P $ and $ Q $ are permutations of
each other.
This is easy to see because only in that case can one have $ H_{X|Y}(S) = H_{Y|X}(S) = 0 $,
for some $ S \in \mathcal{C}(P,Q) $.
Therefore, if distributions are identified up to a permutation, then $ \underline{\Delta}_p $
is a metric.
In other words, if we think of distributions as unordered multisets of nonnegative
numbers summing up to one, then $ \underline{\Delta}_p $ is a metric on such a space.
\end{remark}

Observe that the distribution defining $\underline{\Delta}_p(P,Q)$ is in fact the 
minimum entropy coupling. 
Thus minimum entropy couplings define the distances $\underline{\Delta}_p$ on the space 
of probability distributions in the same way as the maximal coupling defines the total variation 
distance. However, there is a sharp difference in the computational complexity of finding 
these two couplings, as illustrated in the previous section.

\subsection{Some properties of entropy metrics}

Note that $\underline{\Delta}_p$ is a monotonically nonincreasing function of 
$p$. In the following, we will mostly deal with $\underline{\Delta}_1$ and $\underline{\Delta}_\infty$, 
but most results concerning bounds and convergence can be extended to all $\underline{\Delta}_p$ 
based on this monotonicity property.

The metric $\underline{\Delta}_1$ gives an upper bound on the entropy difference 
$|H(P)-H(Q)|$. Namely, since
\begin{equation}
 \begin{aligned}
  |H(X)-H(Y)|&=|H(X|Y)-H(Y|X)|   \\ 
             &\leq H(X|Y)+H(Y|X)   \\
             &=\Delta_1(X,Y), 
 \end{aligned}
\end{equation}
we conclude that
\begin{equation}
\label{deltaent}
  |H(P)-H(Q)|\leq\underline{\Delta}_1(P,Q). 
\end{equation}
Therefore, entropy is continuous with respect to this pseudometric, i.e., 
$\underline{\Delta}_1(P_n,P)\to0$ implies $H(P_n)\to H(P)$. Bounding the entropy 
difference is an important problem in various contexts and it has been studied 
extensively, see for example \cite{ho2,sason}. In particular, \cite{sason} studies 
bounds on the entropy difference via maximal couplings, whereas \eqref{deltaent} 
is obtained via minimum entropy couplings. 

Another useful property, relating the entropy metric $\underline{\Delta}_1$ and 
the total variation distance, follows from Fano's inequality
\begin{equation}
  H(X|Y) \leq \mathbb{P}(X\neq Y)\log(|X|-1) + h(\mathbb{P}(X\neq Y)), 
\end{equation}
where $|X|$ denotes the size of the support of $X$, and 
$h(x)=-x\log_2(x) - (1-x)\log_2(1-x)$, $x\in[0,1]$, is the binary entropy function. 
Evaluating the rhs at the maximal coupling (the joint distribution which 
minimizes $\mathbb{P}(X\neq Y)$), and the lhs at the minimum entropy coupling, we 
obtain
\begin{equation}
  \underline{\Delta}_1(P,Q) \leq d_\textsc{tv}(P,Q)\log(|P||Q|) + 2h(d_\textsc{tv}(P,Q)). 
\end{equation}
This relation makes sense only when the alphabets (supports of $P$ and $Q$) are 
finite. When the supports are also fixed it shows that $\underline{\Delta}_1$ is 
continuous with respect to $d_\textsc{tv}$, i.e., that $d_\textsc{tv}(P_n,P)\to0$ implies 
$\underline{\Delta}_1(P_n,P)\to0$. By the Pinsker-Csisz\'ar-Kemperman inequality
\cite{csiszar}
\begin{equation}
\label{pinsker}
  D(P_n||P) \geq \frac{ 2 }{ \ln 2 } d_\textsc{tv}^2(P_n, P)
\end{equation}
it follows that $\underline{\Delta}_1$ is also continuous with respect to information 
divergence, i.e., $D(P_n||P)\to0$ implies $\underline{\Delta}_1(P_n,P)\to0$. 

The continuity of $\underline{\Delta}_1$ with respect to $d_\textsc{tv}$ fails in 
the case of infinite (or even finite, but unbounded) supports, which follows from 
\eqref{deltaent} and the fact that entropy is a discontinuous functional with respect 
to the total variation distance. One can, however, claim the following. 
\begin{proposition}
 If $P_n\to P$ in the total variation distance, and $H(P_n)\to H(P)<\infty$, then 
 $\underline{\Delta}_1(P_n,P)\to0$. 
\end{proposition}
\begin{proof}
 In \cite[Thm 17]{hoverdu} it is shown that if $d_\textsc{tv}(P_{X_n},P_{X})\to0$ 
 and $H(X_n)\to H(X)<\infty$, then $\mathbb{P}(X_n\neq Y_n)\to0$ implies $H(X_n|Y_n)\to0$, 
 for any r.v.'s $Y_n$. Our claim then follows by specifying $P_{X_n}=P_n$, $P_{X}=P_{Y_n}=P$, 
 and taking infimums on both sides of the implication. 
\end{proof}

It should be pointed out that sharper bounds than the above can be obtained
by using  $\underline{\Delta}_\infty$ instead of $\underline{\Delta}_1$.
For example
\begin{equation}
  |H(P)-H(Q)|\leq\underline{\Delta}_\infty(P,Q), 
\end{equation}
(with equality whenever the minimum entropy coupling of $P$ and $Q$ is such that
$Y$ is a function of $X$, or vice versa), and
\begin{equation}
  \underline{\Delta}_\infty(P,Q)\leq d_\textsc{tv}(P,Q)\log\max\{|P|,|Q|\} + h(d_\textsc{tv}(P,Q)) .
\end{equation}

We conclude this section with an interesting remark on the conditional entropy. 
First observe that the pseudometric $ \Delta_p $ ($ \underline{\Delta}_p $) can
also be defined for random vectors (multivariate distributions).
For example, $ \Delta_1((X,Y),(Z)) $ is well-defined by $ H(X,Y|Z) + H(Z|X,Y) $.
If the distributions of $ (X,Y) $ and $ Z $ are $ S $ and $ R $, respectively,
then minimizing the above expression over all tri-variate distributions with the
corresponding marginals $ S $ and $ R $ would give $ \underline{\Delta}_1(S,R) $.
Furthermore, random vectors can even overlap.
For example, we have
\begin{equation}
  \Delta_1((X),(X,Y)) = H(X|X,Y) + H(X,Y|X) = H(Y|X), 
\end{equation}
because the first summand is equal to zero.
Therefore, the conditional entropy $ H(Y|X) $ can be seen as the distance between
the pair $ (X, Y) $ and the conditioning random variable $ X $.
If the distribution of $ (X, Y) $ is $ S $, and the marginal distribution of $ X $
is $ P $, then
\begin{equation}
  \underline{\Delta}_1(P,S) = H_{Y|X}(S) ,
\end{equation}
because $ S $ is the only distribution consistent with these constraints.
In fact, we have $ \underline{\Delta}_p(P,S) = H_{Y|X}(S) $ for all $ p \in [1, \infty] $.
Therefore, \emph{the conditional entropy $ H(Y|X) $ represents the distance between
the joint distribution of $ (X,Y) $ and the marginal distribution of the conditioning
random variable $ X $}.

\section{Conclusions}

We have presented an information-theoretic view on probability distributions
with fixed margin\-als.
This well-studied topic still provides many interesting research problems and
enables an interplay of several different fields.
Various optimization problems associated with information measures over such sets
of distributions were analyzed and shown to be intractable.
Continuity questions and the existence of extrema of these functionals were also
addressed (in the case of countably infinite alphabets).
A family of information-theoretic pseudometrics was defined and their properties
and relations to other metrics investigated.
A central notion that was introduced in the paper and that represents a connecting
point of the above-mentioned results is the minimum entropy coupling;
the relevance of this notion was demonstrated in several respects.

\section*{Acknowledgments}

The authors would like to thank the reviewers for carefully reading the
manuscript and for providing many suggestions on how to improve it.

\bibliographystyle{model1-num-names}

\end{document}